\documentclass[a4paper,10pt]{article}
\usepackage{amsmath}
\usepackage{amsthm}
\usepackage{amssymb}
\usepackage{latexsym}
\usepackage{multirow}
\usepackage{dsfont}
\usepackage{pdflscape}
\usepackage{lscape}
\usepackage{comment}

\newcommand{\diag}{\textrm{diag}}

\newtheorem{theorem}{Theorem}
\newtheorem{definition}[theorem]{Definition}

\newtheorem{proposition}[theorem]{Proposition}
\newtheorem{corollary}[theorem]{Corollary}

\newtheorem{lemma}[theorem]{Lemma}

\newtheorem{remark}[theorem]{Remark}
\newtheorem{conjecture}[theorem]{Conjecture}

\title{Structure Theorem of Square Complex Orthogonal Design}
\author{Yuan Li}
\date{}

\begin{document}

\maketitle

\begin{abstract}
Square COD (complex orthogonal design) with size $[n, n, k]$ is an $n \times n$ matrix $\mathcal{O}_z$, where each entry is a complex linear combination of
$z_i$ and their conjugations $z_i^*$, $i=1,\ldots, k$,  such that $\mathcal{O}_z^H \mathcal{O}_z = (|z_1|^2 + \ldots + |z_k|^2)I_n$. Closely following the work of
Hottinen and Tirkkonen, which proved an upper bound of $k/n$ by making a crucial observation between square COD and group representation, we prove the structure theorem of square COD.

\end{abstract}

\section{Introduction}
COD (complex orthogonal design) $\mathcal{O}_z$ with size $[p, n, k]$ is a $p \times n$ matrix where each entry
is a complex linear combination of $z_i, z^*_i$, $i=1,\ldots, k$, such that $\mathcal{O}_z^H \mathcal{O}_z = (|z_1|^2+\ldots+|z_k|^2)I_n$.
It is called a square COD if $p = n$. The general question is, for which $p, n, k$, there exists a $[p, n, k]$ complex orthogonal design, and what are the possible
structures?

For square COD, that is $p = n$, Hottinen and Tirkkonen \cite{TH02} proved an upper bound $\frac{\lceil \log_2 n \rceil + 1}{2^{\lceil \log_2 n \rceil}}$ of $k/n$ by making a crucial connection between square COD and group representation.
In fact, closely following their work, the structure of square COD can be clarified, which is what we did in this paper.
In \cite{Lia03}, Liang observed that $[n, n, k]$ square COD exists if and only if $2^{k-1} | n$. Because $[n, n, k]$ square COD induces a family of $2k$ matrices in $GL_n(\mathbb{C})$ such
that any not-all-zero complex linear combinations is nonsingular, and it is known that the size of such family of matrices is bounded
by $2m+2$, where $n = 2^m n_0$, $n_0$ odd, which is a highly nontrivial result proved by Frank Adams, Lax and Phillips in \cite{Adams62}, \cite{ALP65}, \cite{ALP65_2}.

For nonsquare COD, little is known except some restricted cases. In \cite{WX03}, Wang and Xia proved that $k/p \le 3/4$ when $n$ is greater than $2$.
In \cite{Lia03}, when linear combination is not allowed, i.e., each entry is $\pm z_i, \pm z^*_i$ or $0$, Liang proved $k/p \le (m+1)/(2m)$ for $n = 2m$ or $2m-1$, which is tight. In \cite{AKP07}, \cite{AKM10},
S. S. Adams, Karst, Murugan, and Pollack proved tight lower bound of $p$ when $k/p$ reaches the maximal for CODs without linear combinations.
By putting a further restriction that submatrices $\begin{pmatrix} \pm z_j & 0
\\ \ 0 & \pm z^*_j
\end{pmatrix}$ and $\begin{pmatrix} \pm z^*_j & 0 \\ \ 0 & \pm z_j
\end{pmatrix}$ are forbidden, Kan and Li determined all possible $[p, n, k]$ and the structures \cite{LK12}.

A lot of investigation of COD is motivated by Space-time Block Codes (STBC) in wireless communication systems with multiple transmit and receive antennas. Since the pioneering
work by Alamouti \cite{Ala98} in 1998, and the work by Tarokh et al. \cite{TJC99}, \cite{TJC00}, complex orthogonal designs have become an effective technique for the design of STBC. For STBC, parameter $k/p$ is the encoding rate, and $p$ is the decoding delay, that is why the upper bound of $k/p$ and the lower bound of $p$ attracts a lot of attention.
 Because of its applications in space-time block codes, quite a lot of constructions have been proposed \cite{ADK11}, \cite{DR12}, \cite{KS05},  \cite{Lia03}, \cite{LK12}, \cite{LFX05}, \cite{SXL04}, \cite{SSW05}, \cite{TWS09}.

In this paper, we prove the structure theorem of square COD, which roughly says every $[n, n, k]$ square COD
is equivalent to some simple canonical form. We emphasize that although the structure theorem is a satisfying result
describes all possible square CODs, we did little to get it. Nearly all ingredients for the proof are already there, including classical result on representation
of finite groups, and the connection between square COD and group representation in \cite{TH02}.

\section{Preliminaries of Group Representation}

In this section, we review some basic definitions and results on representation of finite groups, which will be used in the following sequel.
The missing proofs can be found in group representation textbooks, for example \cite{Serre}.

A representation $\rho$ of a group $\mathcal{G}$ of dimension $n$ is a homomorphism from $\mathcal{G}$ to $GL_n(\mathbb{C})$, that is,
$$
\rho(g_1) \rho(g_2) = \rho(g_1 g_2)
$$
for any $g_1, g_2 \in \mathcal{G}$. Call $\rho$ is an unitary representation if $\rho$ is a map from $\mathcal{G}$ to $U_n(\mathcal{C})$,
where $U_n(\mathbb{C})$ denotes the group of $n \times n$ unitary matrices.

Two representations $\rho, \pi : \mathcal{G} \to GL_n(\mathbb{C})$ are equivalent (equal) if there exists $T \in GL_n(\mathbb{C})$ such that
$\rho = T \pi T^{-1}$. They are unitarily equivalent if there exists $T \in U_n(\mathbb{C})$ such that $\rho = T \pi T^{-1}$.

Given representation $\rho: \mathcal{G} \to GL_n(\mathbb{C})$, a subspace $V$ of $\mathbb{C}^n$ is called an invariant subspace if
$\rho(g) v \in V$ for any $g \in \mathcal{G}$, any $v \in V$. Representation $\rho$ is called an irreducible representation if $\rho$ does
not have nontrivial invariant subspace (except $0$ and $\mathbb{C}^n$).

For finite groups, any representation is (equivalent to) a direct sum of irreducible ones (unique up to ordering). If finite group $\mathcal{G}$ is explicitly given, it's usually not
difficult to find all irreducible representations. There are two nice counting formulas, which are useful in classifying all irreducible representations: the number of all irreducible
representations equals the number of conjugacy classes; the sum of squares of the dimension of all  irreducible representations
equals the size of the group.

For representation $\rho : \mathcal{G} \to GL_n(\mathbb{C})$, the character $\chi : \mathcal{G} \to \mathbb{C}$
is defined by the trace of the matrix, i.e., $\chi(g) = \mathrm{Tr}(\rho(g))$ for $g \in \mathcal{G}$. The characters of all irreducible
representations form a basis of class functions on $\mathcal{G}$, where a function from $\mathcal{G}$ to $\mathbb{C}$ is a class function if
it takes the same value on every conjugacy class. As a consequence, two representations are equal if and only if their characters are the same.

Let $\rho : \mathcal{G} \to GL_n(\mathbb{C})$ be an irreducible representation of group $\mathcal{G}$. If $T \in GL_n(\mathbb{C})$ intertwines (commutes)
with $\rho$, that is, $T \rho(g) = \rho(g) T$, for all $g \in \mathcal{G}$, Schur's lemma says, $T = \lambda I$ for some $\lambda \in \mathbb{C}$.
\section{Structure Theorem of square COD}

\begin{definition} \cite{TJC99}
 Complex Orthogonal Design (COD) with size $[p, n, k]$ is a $p \times n$ matrix $\mathcal{O}_z$, where each entry is a complex linear combination of
$z_i, z^*_i$, $i=1,\ldots, k$, such that
\begin{equation}
\label{def:COD}
\mathcal{O}_z^H \mathcal{O}_z = (|z_1|^2 + \ldots + |z_k|^2)I_n.
\end{equation}
If $p = n$, it is called a square COD.
\end{definition}
\begin{remark}
In the definition, $z_1, \ldots, z_k$ are indeterminates over $\mathbb{C}$. There are two ways to think of it: $\mathcal{O}_z$ is an unitary matrix
for every $z_1, \ldots, z_k \in \mathbb{C}$ with $|z_1|^2 + \ldots + |z_k|^2 = 1$; or $z_1, \ldots, z_k$ are ``formal'' complex variables such that
\eqref{def:COD} is satisfied.
\end{remark}

Why do we need conjugation $z^*_i$ in the definition of square COD? What if each entry is just complex linear combination of $z_i, i = 1, \ldots, k$? It's not difficult
to see under this definition, there does not exists square COD with $k$ is greater than $1$, and we leave the verification to interested readers.

\vspace{0.2cm}
Assume $\mathcal{O}_z$ is an $[n, n, k]$ square COD, and $U, V \in U_n(\mathbb{C})$, then $U \mathcal{O}_z V$ is also an $[n, n, k]$ square COD, because
\begin{eqnarray*}
(U\mathcal{O}_zV)^H (U\mathcal{O}V) & = & V^H \mathcal{O}_z^H U^H U\mathcal{O}_zV \\
& = & V^H \mathcal{O}_z^H \mathcal{O}_z V \\
& = & V^H (|z_1|^2 + \ldots + |z_k|^2)I_n V\\
& = & (|z_1|^2 + \ldots + |z_k|^2) I_n.
\end{eqnarray*}
Say square CODs $\mathcal{O}_z$ and $U \mathcal{O}_z V$ are equivalent, which defines an equivalence relation among square CODs.

Before stating our main result, we need to define the canonical square CODs.

\begin{definition}
\label{def:can}
Let
\begin{equation}
\mathcal{C}_1 = \begin{pmatrix}
z_1
\end{pmatrix}
\end{equation}
For $k > 1$, let
\begin{equation}
\mathcal{C}_k = \begin{pmatrix}
\mathcal{C}_{k-1} & z_k I_{2^{k-2}} \\
- z^*_k I_{2^{k-2}} & \mathcal{C}^H_{k-1}\\
\end{pmatrix}.
\end{equation}
Define $\mathcal{C}^-_k = \mathcal{C}_k(z_1, \ldots, z_{k-1}, z_k^*)$ be the design by replacing $z_k$ by $z^*_k$ in $\mathcal{C}_k$.
\end{definition}
\begin{remark} Since $z_1, \ldots, z_k$ are totally symmetric in $\mathcal{C}_k$, it doesn't matter which $z_i$ to conjugate in $\mathcal{C}_k^-$.
In other words, we can define $$\mathcal{C}_k^- = \mathcal{C}_k(z_1, \ldots, z_{i-1}, z_i^*, z_{i+1}, \ldots, z_k)$$ for any $i \in [k]$.
\end{remark}

Let's verify $\mathcal{C}_k$ and $\mathcal{C}^-_k$ are square CODs by induction. Since $|z_k|^2 =|z^*_k|^2$, it suffices to prove
$\mathcal{C}_k$ is square COD. When $k=1$, it's obvious. Assuming $\mathcal{C}_{k-1}$ is COD, 
\begin{eqnarray*}
\mathcal{C}^H_k \mathcal{C}_k & = & \begin{pmatrix}
\mathcal{C}_{k-1} & z_k I_{2^{k-2}} \\
- z^*_k I_{2^{k-2}} & \mathcal{C}^H_{k-1}\\
\end{pmatrix}^H \begin{pmatrix}
\mathcal{C}_{k-1} & z_k I_{2^{k-2}} \\
- z^*_k I_{2^{k-2}} & \mathcal{C}^H_{k-1}\\
\end{pmatrix} \\
& = & \begin{pmatrix}
\mathcal{C}^H_{k-1} \mathcal{C}_{k-1} + |z_k|^2 I_{2^{k-2}} & 0  \\
0 & \mathcal{C}^H_{k-1} \mathcal{C}_{k-1} + |z_k|^2 I_{2^{k-2}}\\
\end{pmatrix} \\
& = & \begin{pmatrix}
(|z_1|+\ldots+ |z_k|^2) I_{2^{k-2}} & 0  \\
0 & (|z_1|^2 +\ldots+ |z_k|^2) I_{2^{k-2}}\\
\end{pmatrix} \\
& = & (|z_1|^2 +\ldots+ |z_k|^2) I_{2^{k-1}},
\end{eqnarray*}
where the last second step is by induction hypothesis that $\mathcal{C}^H_{k-1} \mathcal{C}_{k-1} =(|z_1|^2 +\ldots+|z_{k-1}|^2)I_{2^{k-2}}$.

Now we are ready to state our main result, which is the structure theorem of square COD.

\begin{theorem}
\label{thm:main}
 Square COD $\mathcal{O}_z$ with size $[n, n, k]$ exists if and only if
$$
2^{k-1} | n,
$$
and there exist $U, V \in U_n(\mathbb{C})$, and unique $n_1, n_2 \in \mathbb{N}$ with $n_1 + n_2 = n/2^{k-1}$ such that
\begin{eqnarray*}
\mathcal{O}_z & = & U  \diag(\underbrace{\mathcal{C}_k, \ldots, \mathcal{C}_k}_{n_1}, \underbrace{\mathcal{C}^{-}_k, \ldots, \mathcal{C}^{-}_k}_{n_2}) V \\
& = & U \left((I_{n_1} \otimes \mathcal{C}_k) \oplus (I_{n_2} \otimes \mathcal{C}^-_k)\right) V.
\end{eqnarray*}
\end{theorem}

At first glance, you may doubt the correctness of the above theorem. Let $\mathcal{O}_z$ be some square COD by arbitrarily conjugating
some $z_i$ from the canonical form $\mathcal{C}_k$. Is $\mathcal{O}_z$ still equivalent to $\mathcal{C}_k$ or $\mathcal{C}_k^-$? Of course, the answer is
yes. Because $\mathcal{C}_k$ has high ``symmetry'' by its construction: $2^{k-1}$ rows contain $2^{k-1}$ different conjugation patterns. By ``conjugation pattern'', we mean a set $S \subseteq [k]$, such that $\pm z_i^{*}$ instead of $\pm z_i$ appears in that row if and only if $i \in S$.

\vspace{0.2cm}
Before starting our proof, let's sketch the overall idea. First, the existence of $[n, n, k]$ square COD implies the existence of a set of matrices $E_0, \ldots, E_{2k-1} \in U_n(\mathbb{C})$ such that, for any $i \not= j$,
\begin{equation}
\label{equ:Eij}
E^H_i E_j + E^H_j E_i = 0.
\end{equation}
Following a normalization step made in \cite{TH02}, define $G_i = E^H_0 E_i$, which is also unitary, then \eqref{equ:Eij} implies
\begin{equation}
G_iG_j = -G_j G_i
\end{equation}
for all $i \not= j$, and $G^2_i =-1$ for all $i \in \{2, \ldots, 2k\}$.  If we artificially define a group generated by $g_2, \ldots, g_{2k}$ as well as
$1, -1$,
satisfying relations $g_i^2 = -1$ and $g_i g_j = -g_j g_i$, then matrices $G_2, \ldots, G_{2k}$ induce a linear representation of the group in
the natural way.
(In fact, this is the defining relation of generators of Clifford algebra, which has been well studied in mathematics.) By studying the structure of
the group, it turns out that this group has only two nondegenerate irreducible representations, which are exactly induced by $\mathcal{C}^\pm_k$.
Since any linear representation of finite group
can be decomposed into a direct sum of irreducible ones, we obtain our structure theorem.

\vspace{0.2cm}
Let's start our proof formally. Given an $[n, n, k]$ square COD $\mathcal{O}_z$, writing $z_i = x_i + \sqrt{-1}y_i$, $x_i, y_i \in \mathbb{R}$, expand $\mathcal{O}_z$ as follows.
\begin{eqnarray*}
\mathcal{O}_z & = & \sum_{i=1}^k \left( z_i A_i + z_i^* B_i \right)\\
& = & \sum_{i=1}^k \left( (x_i + \sqrt{-1}y_i)A_i + (x_i - \sqrt{-1}y_i)B_i \right)\\
& = & \sum_{i=1}^k x_i (A_i + B_i) + \sum_{i=1}^k y_i \sqrt{-1} (A_i - B_i).
\end{eqnarray*}
For convenience, let $E_i = A_i + B_i$, $E_{i+k} = \sqrt{-1} (A_i - B_i)$ and $x_{k+i} = y_i$. Then,
\begin{equation}
\label{equ:UEi}
\mathcal{O}_z = \sum_{i=1}^{2k} x_i E_i,
\end{equation}
and
\begin{eqnarray*}
\mathcal{O}^H_z & = & \sum_{i = 1}^k (z_i^* A_i^H + z_i B_i^H) \\
& = & \sum_{i = 1}^k \left((x_i - \sqrt{-1} y_i) A_i^H + (x_i + \sqrt{-1} y_i) B_i^H \right)\\
& = & \sum_{i=1}^k x_i (A_i + B_i)^H - \sum_{i=1}^k y_i \sqrt{-1} (A_i - B_i)^H \\
& = & \sum_{i = 1}^{2k} x_i E_i^H.
\end{eqnarray*}
By taking $x_i = 1$ and all others $0$, condition $\mathcal{O}_z^H \mathcal{O}_z = (|z_1|^2 + \ldots + |z_k|^2)I_n$ implies
\begin{equation}
\label{equ:orth}
E^H_i E_i = I_n.
\end{equation}
By taking $x_i = x_j = 1$ and all others $0$, we have
\begin{equation}
\label{equ:van}
 E_i^H E_j + E_j^H E_i = 0.
\end{equation}
On the other hand, if \eqref{equ:orth} and \eqref{equ:van} are satisfied, it's easy to check $\mathcal{O}_z^H \mathcal{O}_z =
 (\sum_{i=1}^{2k} x_i^2) I_n = (\sum_{i=1}^k |z_i|^2) I_n$,
which means \eqref{equ:orth} and \eqref{equ:van} are both necessary and sufficient.
Now, we have proved the following proposition, which is a folklore result. We are not sure who prove it first, at least it appears
in \cite{TH02}. In the pioneering work \cite{TJC99} which first defines COD motivated by Space-time Block Codes, it seems that they don't get this concise form, which is crucial for the following work.

\begin{proposition}
\label{prop:basic}
Square COD with size $[n, n, k]$ exists if and only if there exists matrices $E_1, \ldots, E_{2k} \in U_n(\mathbb{C})$, such that
$$
 E_i^H E_j + E_j^H E_i = 0
$$
for all $1 \le i \not= j \le 2k$.
\end{proposition}
\begin{remark} The above proposition also holds for non-square case, that is, COD with size $[p, n, k]$ exists if and only if there exists
matrices $E_1, \ldots, E_{2k} \in M_{p \times n}(\mathbb{C})$ such that
$$
E_i^H E_i = I_n
$$
for all $i$, and
$$
E_i^H E_j + E_j^H E_i = 0
$$
for all $i \not= j$.
\end{remark}

For convenience of description, let's left shift the indices of $E_i$ by 1.
Define $G_i = E_0^H E_i$. Then $G_i^H G_i = (E_0^H E_i)^H (E_0^H E_i) = E_i^H E_0 E_0^H E_i = I_n$, which means $G_i$ is also unitary. Further,
$G_i$ is skew-Hermitian (or equivalently, $G_i^2 = -I_n$), whence
$$
G_i^H = ((E_0)^H E_i)^H = E_i^H E_0 = -E^H_0 E_i = -G_i,
$$
where the last second step is from $E_0^H E_i + E_i^H E_0 = 0$. Further,
we have
\begin{eqnarray*}
&  & G_i G_j =  -G^H_i G_j\\
& = & -(E_0^H E_i)^H (E_0^H E_j) =  -E^H_i E_0 E_0^H E_j \\
& = & -E^H_i E_j = E_j^H E_i = -(-E_j^H E_i) \\
& = & -G_j G_i,
\end{eqnarray*}
which means $G_i$ and $G_j$ are anti-commuting.

Now, let's artificially define a group $\mathcal{G}_{2k-1}$ generated by $g_1, \ldots, g_{2k-1}, -1$ satisfying $
g_i^2 = -1
$
and
$
g_i g_j = -g_j g_i.
$
Notice that $1$ and $-1$ denote two distinct elements in the group, where $1$ is the identity, and
$g_i, -g_i$ are two different elements satisfying $-g_i = (-1) g_i$. Formally, the group consists of the following
elements
$$
\{ \pm \prod_{i \in S} g_i : S \subseteq \{1,\ldots, 2k-1\}\}.
$$
Thus, the size of the group $\mathcal{G}_{2k-1}$ is $2^{2k}$.

Lemma \ref{lem:irr_rep} and Lemma \ref{lem:groupstr} are about the irreducible representations of the group
$\mathcal{G}_{2k-1}$, where \cite{TH02} contains a proof, and we reproduce the proof in the appendix for completeness.

\begin{lemma}
\label{lem:irr_rep}
For group $\mathcal{G}_{2k-1}$, if $\rho$ is an irreducible representation with dimension greater than $1$, then
$\pi$ defined by $\pi(g_i) = \rho(g_i)$ for all $i \not= 2k-1$, and $\pi(g_{2k-1}) = -\rho(g_{2k-1})$ is another
irreducible
representation.
\end{lemma}

Next lemma shows that there are only two nondegenerate irreducible representations of $\mathcal{G}_{2k-1}$, and
both of them are of dimension $2^{k-1}$.

\begin{lemma}
\label{lem:groupstr}
Group $\mathcal{G}_{2k-1}$ has $2^{2k-1}+2$ irreducible representations. Two are $2^{k-1}$ dimensional,
$2^{2k-1}$ are 1 dimensional.
\end{lemma}

In fact, we can write down all the irreducible representations explicitly, for example, see \cite{TH02}. However, we could avoid doing that.

\begin{lemma} Square CODs $\mathcal{C}_{k}$ and $\mathcal{C}^-_{k}$ induces two nonequivalent $2^{k-1}$ dimensional
irreducible
representations.
\end{lemma}
\begin{proof} Denote by $\rho$ and $\rho'$ the group representations induced by $\mathcal{C}_{k}$ and $\mathcal{C}^-_{k}$.
Since $\dim(\rho) = 2^{k-1}$, by Lemma \ref{lem:groupstr}, $\rho$ is either irreducible or a direct sum of $2^{k-1}$ one dimensional
representations. Since all one dimensional representation are degenerate, and $\rho$ is non-degenerate, it should be a $2^{k-1}$ dimensional
irreducible representation, as well as $\rho'$.

As usual, let $\mathcal{C}_k = \sum_{i=1}^k \left( z_i A_i + z_i^* B_i \right)$, and then $E_i = A_i + B_i$, $E_{i+k} = \sqrt{-1} (A_i - B_i)$.
The only difference between $\mathcal{C}_k$ and $\mathcal{C}^-_k$ is that $z_k$ is conjugated, which results in swapping $A_k$ and $B_k$,
and thus $E_{2k}$ is negated while all the other $E_i$'s are unchanged, i.e., only $G_{2k-1}$ is negated. By Lemma \ref{lem:irr_rep}, we know that $\rho'$ is another
irreducible representation different from $\rho$.
\end{proof}

Before proving the structure theorem, let's prove a lemma about unitary representations of a finite group, which
says if two unitary representations are similar, then they are unitarily similar, in the sense that the linear
transformation is unitary. We feel that this result is very likely to be known in math literature, although we can't find
an exact place where it appears.

\begin{lemma}
\label{lem:unitary_equ}
Let $\pi, \sigma : G \to U_n(\mathbb{C})$ be two unitary representations of finite group $G$. If $\pi, \sigma$ are equivalent, i.e., there exists $T \in GL_n(\mathbb{C})$ such that $$T \pi(g) = \sigma(g) T, \quad\quad \forall g \in G,$$ then $\pi$ and $\sigma$ unitarily equivalent, that is, there exists $T' \in U_n(\mathbb{C})$ such that $$T' \pi(g) = \sigma(g) T', \quad\quad \forall g \in G.$$
\end{lemma}
\begin{proof} Prove by construction. Since $T \pi = \sigma T$ and $\pi$ is unitary, which implies $\pi(g)^H = \pi(g^{-1})$, we have $\pi T^H = (\pi^{-1})^H T^H = (T \pi^{-1})^H = (\sigma^{-1} T)^H = T^H (\sigma^{-1})^H = T^H \sigma$. Thus, $T^H$ also intertwines with $\pi$ and $\sigma$. Define $|T| = \sqrt{T^H T}$, which is meaningful since a positive-semidefinite Hermitian matrix has a unique positive-semidefinite square root.

Letting $T' = T|T|^{-1}$, we claim this is the desired $T'$. First, let's verify that $T'$ is unitary.
\begin{eqnarray*}
T'^H T' & = & (T|T|^{-1})^H T|T|^{-1}\\
 = |T|^{-1} T^H T |T|^{-1}
 & = & |T|^{-1} |T|^2 |T|^{-1}
  = I.
\end{eqnarray*} Then, let's verify $T'$ intertwines,
i.e., $T' \pi = \sigma T'$, which is $T |T|^{-1} \pi = \sigma T |T|^{-1} = T \pi |T|^{-1}$. Since $T$ is invertible,
it suffices to prove $|T|^{-1} \pi = \pi |T|^{-1}$, that is, $|T|$ commutes with $\pi$. Notice that $T^H T$ commutes with $\pi$, since $T^H T \pi = T^H \sigma T = \pi T^H T$, and $|T|$ can be approximated by polynomials in $T^* T$ (For example, apply Weierstrass approximation theorem).
Combining with the fact that every polynomial in $T^H T$ commutes with $\pi$, we have $|T|$ also commutes with $\pi$, which completes our proof.
\end{proof}

Now, we are ready to prove the structure theorem of square COD.

\begin{proof} Let $\mathcal{O}_z$ be an $[n, n, k]$ square COD. As in \eqref{equ:UEi}, splitting the real part
and imaginary part, write
$$
\mathcal{O}_z = x_0 E_0 + x_1 E_1 + \ldots + x_{2k-1} E_{2k-1}.
$$
As usual, let  $G_i = E_0^H E_i$ for $i=1,2,\ldots, 2k-1$. It turns out $G_1, \ldots, G_{2k-1}$ induce a representation $\rho$, which is an unitary, nondegenerate representation of group $\mathcal{G}_{2k-1}$ ($\rho$ is defined in the natural way, i.e., $\rho(g_i) = G_i$ for all $i$, and $\rho(-1) = -I$).

Since every representation of a finite group is a direct sum of irreducible ones (see Section 2), there
exists  $T \in GL_n(\mathbb{C})$ such that
\begin{eqnarray}
 \rho(g_i) & = & T^{-1} \rho_1(g_i) \oplus \ldots \oplus \rho_m(g_i) T \nonumber\\
& = & T^{-1} \begin{pmatrix} \rho_1(g_i) & & & \\
& \rho_2(g_i) & &\\
& & \ddots & \\
& & & \rho_m(g_i) \end{pmatrix} T, \label{equ:Giis}
\end{eqnarray}
where $\rho_1, \ldots, \rho_m$ are irreducible representations of group $\mathcal{G}_{2k-1}$. By Lemma \ref{lem:unitary_equ}, $T$ could be chosen to be unitary.

Next, we shall show that $\rho_1, \ldots, \rho_m$ are all nondegenerate. Otherwise, assume $\rho_1$ is degenerate without loss of generality, i.e., $\rho_1(1) = \rho_1(-1)$. Then,
$$
\rho(1) = G_1 = T^{-1} \begin{pmatrix} \rho_1(1) & & & \\
& \rho_2(1) & &\\
& & \ddots & \\
& & & \rho_m(1)  \end{pmatrix} T
$$
and
$$
\rho(-1) = -G_1 = T^{-1} \begin{pmatrix} \rho_1(-1) = \rho_1(1) & & & \\
& \rho_2(-1) & &\\
& & \ddots & \\
& & & \rho_m(-1) \end{pmatrix} T.
$$
Thus
$$
 \begin{pmatrix} \rho_1(1) & & & \\
& \rho_2(1) & &\\
& & \ddots & \\
& & & \rho_m(1) \end{pmatrix} =
\begin{pmatrix} - \rho_1(1) & & & \\
& -\rho_2(-1) & &\\
& & \ddots & \\
& & & -\rho_m(-1) \end{pmatrix},
$$
which implies $\rho_1(1) = - \rho_1(1) \Rightarrow \rho_1(1) = 0 \Rightarrow \rho_1 = 0$. Contradiction!

By Lemma \ref{lem:groupstr}, all non-degenerate irreducible representations are of dimension $2^{k-1}$, we have $\dim(\rho_i) = 2^{k-1}$ for $i = 1, \ldots, m$, which implies $n=m 2^{k-1}$ for some integer $m$, which proves the first part of the theorem.

\vspace{0.2cm}
In order to prove the second part, we will expand $\mathcal{O}_z$ explicitly by \eqref{equ:Giis}.
By definition, $G_i = E_0^H E_i \Rightarrow E_i = E_0 G_i$. By \eqref{equ:Giis},
$$
E_i = E_0 T^{-1} \begin{pmatrix} \rho_1(g_i) & & & \\
& \rho_2(g_i) & &\\
& & \ddots & \\
& & & \rho_m(g_i) \end{pmatrix}T.
$$
Then,
\begin{eqnarray*}
\mathcal{O}_z & = & x_0 E_0 + x_1 E_1 + \ldots + x_{2k-1} E_{2k-1} \\
& = & x_0 E_0 + \sum_{i=1}^{2k-1} x_i E_0 T^{-1} \begin{pmatrix} \rho_1(g_i) & & & \\
& \rho_2(g_i) & &\\
& & \ddots & \\
& & & \rho_m(g_i) \end{pmatrix} T\\
& = &  E_0T^{-1} (Ix_0) T + E_0 T^{-1} \begin{pmatrix} \sum_{i=1}^{2k-1} x_i \rho_1(g_i) & & \\
& \ddots & \\
& & \sum_{i=1}^{2k-1} x_i \rho_m(g_i) \end{pmatrix} T\\
& = & E_0 T^{-1} \begin{pmatrix} x_0I + \sum_{i=1}^{2k-1} x_i \rho_1(g_i) & & \\
& \ddots & \\
& & x_0I + \sum_{i=1}^{2k-1} x_i \rho_m(g_i) \end{pmatrix} T\\
& = & E_0 T^{-1} \begin{pmatrix} \mathcal{C}_k^{\pm} & & \\
& \ddots & \\
& & \mathcal{C}_k^{\pm} \end{pmatrix} T.
\end{eqnarray*}
Set $U = E_0 T^{-1}$ and $V = T$.

Without loss of generality, assume
$$
\mathcal{O}_z =
U  \diag(\underbrace{\mathcal{C}_k, \ldots, \mathcal{C}_k}_{n_1}, \underbrace{\mathcal{C}^{-}_k, \ldots, \mathcal{C}^{-}_k}_{n_2}) V ,
$$
for some $n_1 + n_2 = n/2^{k-1}$, which can be achieved by permuting $U$ and $V$. Thus, representation
$\rho$ induced by $\mathcal{O}_z$ is a direct sum of $n_1$ copies of $\rho_1$ and $n_2$ copies of $\rho_2$, where $\rho_1, \rho_2$ are induced by $\mathcal{C}_k, \mathcal{C}^-_k$
respectively, which implies $\mathrm{Tr}(\rho) = n_1 \mathrm{Tr}{(\rho_1)} +  n_2 \mathrm{Tr}{(\rho_2)}$. Since $\mathrm{Tr}(\rho_1), \mathrm{Tr}(\rho_2)$
are linearly independent (recall that the characters of all irreducible representations form a basis for class functions), $n_1, n_2$ are uniquely determined.
\end{proof}
\begin{remark} By Theorem \ref{thm:main}, for COD with size $[n, n, k]$, there are $n/2^{k-1} + 1$ canonical forms and thus $n/2^{k-1} + 1$  equivalent classes.
\end{remark}

\begin{remark} If $n = 2^{k-1}$, $\rho$ is irreducible, and $U, V$ are uniquely determined by Schur's lemma. If $n = 2^k$ and $\mathcal{O}_z$ is equivalent to $\mathcal{C}_k
\oplus \mathcal{C}^-_k$, $U$ and $V$ are also uniquely determined by a simple generalization of Schur's lemma. For all the other cases, $U$ and $V$ are not unique.
\end{remark}

\section{Square COD without linear combination}
From the main theorem, it's not difficult to prove the following result, which is
the structure theorem for square COD without linear combination, that is, each entry is either $\pm z_i, \pm z_i^*$ or $0$.

\begin{definition}
Square COD without linear combination of size $[n, n, k]$ is a COD such that each
entry is $\pm z_i, \pm z_i^*$ or $0$, $i = 1, \ldots, k$.
\end{definition}

For square COD without linear combination, we have similar conclusion, that any design
can be obtained from canonical one by left multiplying $U$ and right multiplying $V$, where
$U, V \in U_n(\mathbb{C})$ are permutation matrices with signs (nonzero entries is either $1$
or $-1$, and each row and column has only one nonzero entry). The proof is not very difficult given the main theorem. However, we don't know whether it can be proved by combinatorial argument without the main theorem.

\begin{corollary} Square COD of size $[n, n, k]$ without linear combination $\mathcal{O}_z$ exists if and only if
$$
2^{k-1} | n,
$$
and $\mathcal{O}$ can be obtained from
\begin{eqnarray*}
 &  & \diag(\underbrace{\mathcal{C}_k, \ldots, \mathcal{C}_k}_{n_1}, \underbrace{\mathcal{C}^{-}_k, \ldots, \mathcal{C}^{-}_k}_{n_2}) \\
& = & (I_{n_1} \otimes \mathcal{C}_k) \oplus  (I_{n_2} \otimes \mathcal{C}_k^{-})
\end{eqnarray*}
for some $n_1 + n_2 = n/2^{k-1}$,
by row and column permutation, and possibly multiply some rows or columns by $-1$.
\end{corollary}
\begin{proof} The first part, that $[n, n, k]$ square COD without linear combination $\mathcal{O}_z$ exists if and only if
$
2^{k-1} | n,
$ follows directly from our main theorem.

For convenience, let's call the following ``equivalent operations''.
\begin{itemize}
\item Permute rows or columns.
\item Multiply some rows or columns by $-1$.
\end{itemize}
Now, we shall prove by induction on $k$ that $\mathcal{O}_z$ can be obtained from $\bigoplus_{i=1}^{n/2^{k-1}} \mathcal{C}_k^{\pm}$
by equivalent operations. For $k = 1$, canonical form $\mathcal{C}_1^\pm \in \{ (z_1), (z^*_1) \}$, and by equivalent operations, $\mathcal{O}_z$ can be transformed into a diagonal matrix with $z_1$ or $z^*_1$ in its diagonal.

For $k > 1$, since every variable (including $\pm z_i$ or $\pm z^*_i$) appears in each column exactly once, by equivalent operations,  $\mathcal{O}_z$ can be transformed into the following form,
$$
\begin{pmatrix}
\mathcal{O}_z' & z_k I_{n_1} \\
z^*_k I_{n_2} & -\mathcal{O}_z'^H
\end{pmatrix},
$$
where $\mathcal{O}_z'$ is a $[n_1, n_2, k-1]$ COD without linear combinations. If we could prove $n_1 = n_2$, then we are done. Because assume
$\mathcal{O'}_z = \bigoplus_{i=1}^{n/2^{k-1}} \mathcal{C}_{k-1}^{\pm}$, then $\mathcal{O}_z$ can be transformed into the following form by basic operations.
$$
\begin{pmatrix}
\bigoplus_{i=1}^{n/2^{k-1}} \mathcal{C}_{k-1}^{\pm} & z_k I_{n/2} \\
z^*_k I_{n/2} & -\bigoplus_{i=1}^{n/2^{k-1}} \left({\mathcal{C}_{k-1}^{\pm}}\right)^H \\
\end{pmatrix},
$$
which is equivalent to
$\bigoplus_{i=1}^{n/2^{k}} \mathcal{C}_{k}^{\pm}$
by basic operations.

Now, it remains to prove $n_1 = n_2$, and this is where we apply our main theorem. Write
$$
\mathcal{O}_z = \sum_{i = 1}^k z_i A_i + \sum_{i=1}^k z^*_i B_i,
$$
where $A_i, B_i \in M_{n \times n}(\mathbb{C})$. The key observation is: $n_1, n_2$ is the rank of $A_k, B_k$ respectively, because
  the $(i, j)$ entry of $A_k$ and $B_k$
can not be simultaneously nonzero, and $A_k + B_k$ is a permutation matrix with possible signs on rows. Why $A_k + B_k$
is a permutation matrix up to signs? This is because entries in $\{ \pm z_k, \pm z^*_k \}$ should appear in each row
 (and column) exactly once.
 
 By our main theorem, $\mathcal{O}_z = U \mathcal{C} V$,
where $U, V$ are $n \times n$ unitary matrices and $\mathcal{C} = \bigoplus_{i=1}^{n/2^{k}} \mathcal{C}_{k}^{\pm}$
is some canonical form. Similarly write,
$$
\mathcal{C} = \sum_{i = 1}^k z_i A'_i + \sum_{i=1}^k z^*_i B'_i.
$$
Since $\mathcal{O}_z = U \mathcal{C} V$, $A_{k} = U A'_{k} V$ and $B_k = U B'_k V$. It's easy to see the rank of $A'_k$
and $B'_k$ are both $n/2$ by the definition of $\mathcal{C}_k$, which implies the rank of $A_k$ and $B_k$ are also $n/2$, for unitary transformation does
not change the rank.
\end{proof}

\section{COD and Sum of Squares Problem}

Sum of squares composition formula of size $[r, s, n]$ over some field $\mathbb{F}$ is the following identity,
$$
(x_1^2 + \ldots + x_r^2) (y_1^2 + \ldots + y_s^2) = z_1^2 + \ldots + z_n^2,
$$
where each $z_i = z_i(X, Y)$ is a bilinear form in $X$ and $Y$. This problem has been investigated
early in the 19th century (see \cite{Shapiro}, an extensive book on this subject). For general $[r, s, n]$,
this problem is still wide open. For the case $s = n$, Radon and Hurwitz in 1920s proved that formula $[r, n, n]$
over $\mathbb{R}$ or $\mathbb{C}$
exists if and only if $r \le \rho(n)$, where the Hurwitz-Radon function $\rho(n)$ is defined as
\begin{equation}
\label{equ:defrho}
\rho(n) = \begin{cases}
2m+1, & \text{if } m \equiv 0 \pmod 4,\\
2m, & \text{if } m \equiv 1, 2 \pmod 4,\\
2m+2, & \text{if } m \equiv 3 \pmod 4,
\end{cases}
\end{equation}
where $n = 2^m n_0$, $n_0$ odd.

In fact, $[r, s, n]$ formula over $\mathbb{F}$ is equivalent to orthogonal design (OD) over $\mathbb{F}$ with size $[n, s, r]$, where
``OD over $\mathbb{F}$'' is defined as follows.

\begin{definition} \cite{Shapiro} OD (Orthogonal Design) $\mathcal{O}_x$ over $\mathbb{F}$ with size $[p, n, k]$ is a $p \times n$ matrix,
where each entry is an $\mathbb{F}$-linear combination of $x_1, \ldots, x_k$ (think of $x_1, \ldots, x_k$ as formal variables)
such that,
$$
\mathcal{O}_x^T \mathcal{O}_x = (x_1^2 + \ldots + x_k^2) I_n.
$$
\end{definition}
\begin{remark}
Note that OD over $\mathbb{C}$ is not COD. Because in the definition of OD, we take the transpose, while
in COD, it's Hermitian transpose. 
\end{remark}

Recall the definition of COD, $z_1, \ldots, z_k$ are formal complex variables.
In fact, a formal complex variable is equivalent to two formal (real) variables.
For this reason, using formal complex variables seems redundant, and thus we define HOD (Hermitian Orthogonal Design) which
captures COD as a special case.

\begin{definition} HOD (Hermitian Orthogonal Design) $\mathcal{O}_x$ over $\mathbb{F}$ with size $[p, n, k]$ is a $p \times n$ matrix,
where each entry is an $\mathbb{F}$-linear combination of $x_1, \ldots, x_k$
such that,
$$
\mathcal{O}_x^H \mathcal{O}_x = (x_1^2 + \ldots + x_k^2) I_n.
$$
Here, $\mathbb{F}$ is a subring or subfield of $\mathbb{C}$, and when $\mathbb{F}$ is omitted, we assume $\mathbb{F} = \mathbb{C}$.
\end{definition}

By Proposition \ref{prop:basic}, COD with size $[p, n, k]$ is equivalent to HOD with size $[p, n, 2k]$. Since sum of squares formula $[r, s, n]$ is equivalent to OD with size $[n, s, r]$ \cite{Shapiro}, let's forget
about sum of square formulas, and compare what is the difference between OD and HOD.

\vspace{0.2cm}
It is well known \cite{Shapiro} (ch. 0, pp. 3) that OD over field $\mathbb{F}$ with size $[p, n, k]$ is equivalent to $k$ matrices $E_1, \ldots, E_k \in M_{p \times n}(\mathbb{F})$ satisfying
\begin{eqnarray}
E_i^T E_i & = & I_n, \text{ for all } i, \nonumber \\
E_i^T E_j + E_j^T E_i & = & 0, \text{ for all } i \not= j. \label{equ:matrix_OD}
\end{eqnarray}
Similarly, HOD with size $[p, n, k]$ has almost the same characterization except that the transpose $T$ is replaced by Hermitian transpose $H$, that is,
HOD  with size $[p, n, k]$ is equivalent to $k$ matrices $E_1, \ldots, E_k \in M_{p \times n}(\mathbb{C})$ such that
\begin{eqnarray}
E_i^H E_i & = & I_n, \text{ for all } i, \nonumber \\
E_i^H E_j + E_j^T E_i & = & 0, \text{ for all } i \not= j. \label{equ:matrix_HOD}
\end{eqnarray}

It's not clear to us what is the essential difference between \eqref{equ:matrix_OD} and \eqref{equ:matrix_HOD}, but they are surely different. For OD over $\mathbb{R}$ or $\mathbb{C}$, it is known that $[n, n, k]$ is admissible if
and only if $k \le \rho(n)$, where $\rho(n)$ is defined in \eqref{equ:defrho}; For HOD, size $[n, n, 2k]$ is admissible if and only if $2k \le 2m+2$, where $n = 2^m n_0$, $n_0$ odd, by our main theorem. To see another difference, if $k > p$, OD over $\mathbb{C}$ with size $[p, n, k]$ does not exist, since OD with size $[p, k, n]$, $p < k$, does not exist (the fact $[p, n, k]$ is equivalent to $[p, k, n]$ is clear from the view point of Sum of Squares problem, that is, $X$ and $Y$ are symmetric).
However, HOD with size $[p, n, k]$, $p > k$, does exist. Recall that $\mathcal{C}_2$ is COD of size $[2, 2, 2]$, and thus HOD of size $[2, 2, 4]$.
For more such examples, in \cite{Lia03}, Liang constructed  CODs with size $[p, n, k]$ where $k/p = (m+1)/(2m)$, where $n = 2m-1$ or $2m$, which implies the
existence of HODs with size $[p, n, 2k]$, where $2k/p = (m+1)/m > 1$.

\begin{conjecture} \cite{Shapiro}
OD over filed $\mathbb{C}$ with size $[p, n, k]$ is admissible if and only if the same size is admissible over ring $\mathbb{Z}$.
\end{conjecture}
Assuming the above bold conjecture in \cite{Shapiro} (ch. 14.22, pp. 314), the existence of OD over field $\mathbb{C}$ with size $[p, n, k]$ implies the existence of Hermitian OD with size $[p, n, k]$, but the converse is not true. Corresponding to the above conjecture, there is also a similar  bold conjecture for COD.

\begin{conjecture}
\label{conj:COD}
 COD over $\mathbb{C}$ with size $[p, n, k]$ exists if and only if it is admissible for COD over $\mathbb{Z}$, that is, COD without linear
combinations.
\end{conjecture}

 The similar statement for HOD is not true, where $\mathbb{Z}$ should be replaced by Gaussian integers $\mathbb{Z}[i]$. 

We believe HOD (as well as COD as a special case), or equivalently, equations \eqref{equ:matrix_HOD}, is interesting in its own right. As we have already seen, the definition of square COD is especially nice for our proof of the structure theorem, due to the special role of unitary matrices in group representation.
Although HOD captures the definition of COD, the definition of COD still has its our merits, besides its applications in STBC. One reason may be Conjecture \ref{conj:COD}. When restricting our attention to the case $\mathbb{Z}$, there are some interesting results on COD over $\mathbb{Z}$ by combinatorial method \cite{AKP07}, \cite{AKM10}, \cite{Lia03}, \cite{LK12}, and might bring insight for COD $\mathbb{C}$.

Table \ref{tab:OD} summarizes three different definitions, OD (Orthogonal Design), HOD (Hermitian Orthogonal Design)
and COD (Complex Orthogonal Design), as well as their matrix equations characterization.

\begin{table}
\caption{Comparison between OD, HOD and COD}
\label{tab:OD}
\begin{tabular}{ p{3cm} | p{4cm} | p{4cm} }
Name & Definition & Matrices characterization \\
\hline
OD over $\mathbb{F}$ with size $[p, n, k]$, where
$\mathbb{F}$ is a ring or field
($\equiv$ Sum of Squares over $\mathbb{F}$ with size $[k, n, p]$) &
$\mathcal{O}_x$ is an $p\times n$ matrix with each entry
$\mathbb{F}$-linear combination of $x_1, \ldots, x_k$ such that
$\mathcal{O}_x^T \mathcal{O}_x = (x_1^2 + \ldots + x_k^2) I_n$.
& $\exists A_1, \ldots, A_k \in M_{p \times n}(\mathbb{F})$ such that
$A_i^T A_i = I_n \quad \forall i$, and
$A_i^T A_j + A_j^T A_i = 0 \quad \forall i \not= j$.\\
\hline
HOD over $\mathbb{F}$ with size $[p, n, k]$ ($\mathbb{F}$ is a
subring or subfield of $\mathbb{C}$) &
$\mathcal{O}_x$ is an $p \times n$ matrix with each entry $\mathbb{F}$-linear combination
of $x_1, \ldots, x_k$ such that $\mathcal{O}_x^H \mathcal{O}_x = (x_1^2 + \ldots + x_k^2) I_n$. &
$\exists A_1, \ldots, A_k \in M_{p \times n}(\mathbb{F})$ such that
$A_i^H A_i = I_n \quad \forall i$ and $A_i^H A_j + A_j^H A_i = 0 \quad \forall i \not= j$. \\
\hline
COD over $\mathbb{F}$ with size $[p, n, k]$, where $\mathbb{F}$ is a ring or field ($\equiv$ HOD over $\mathbb{F}$ with size $[p, n, 2k]$ when $\sqrt{-1} \in
\mathbb{F}$)  &
$\mathcal{O}_z$ is an $p \times n$ matrix with each entry $\mathbb{F}$-linear combination
of $z_1, \ldots, z_k$ and their conjugations $z_1^*, \ldots, z_k^*$ such that $\mathcal{O}_z^H \mathcal{O}_z = (|z_1|^2 + \ldots + |z_k|^2) I_n$. &
$\exists E_1, \ldots, E_k \in M_{p \times n}(\mathbb{F})$ and $\exists E_{k+1}, \ldots, E_{2k} \in M_{p \times n}(\mathbb{F}) \sqrt{-1}$ such that
$E_i^H E_i = I_n \quad \forall i$ and $E_i^H E_j + E_j^H E_i = 0 \quad \forall i \not= j$.
\end{tabular}
\end{table}

Since COD is a special case of HOD with $k$ even, it's natural to ask what is the structure for HOD when $n$ is odd, which is the result of the following
section.


\section{Square HOD}

In this section, we will prove the structure theorem for square HOD by the same group representation method. Note that $[n, n, k]$ square COD is equivalent to
$[n, n, 2k]$ square HOD, and we already proved the structure theorem for square COD in Section 3. It remains to prove the structure theorem for square
HOD when $k$ is odd.

Let's define the canonical form first. Let $\mathcal{H}_{2k}$ be $\mathcal{C}_k$ in Definition \ref{def:can} by replacing $z_i$ by $x_{2i-1} + x_{2i} \sqrt{-1}$
and replacing $z_i^*$ by $x_{2i-1} -  x_{2i} \sqrt{-1}$, that is,
$$
\mathcal{H}_{2} = \begin{pmatrix} x_1 +  x_2 \sqrt{-1} \end{pmatrix},
$$
and
$$
\mathcal{H}_{2k} = \begin{pmatrix}
\mathcal{H}_{2k-2} & (x_{2k-1} + x_{2k} \sqrt{-1}) I_{2^{k-2}} \\
- (x_{2k-1} - x_{2k} \sqrt{-1}) I_{2^{k-2}} & \mathcal{H}^H_{2k-2}\\
\end{pmatrix}
$$
for every $k > 1$.

Let $\mathcal{H}_{2k}^- = \mathcal{H}_{2k}(x_1, \ldots, x_{2k-1}, -x_{2k})$, which is obtained from $\mathcal{H}_{2k}$ by negating $x_{2k}$.
$\mathcal{H}_{2k}^\pm$ denotes either $\mathcal{H}_{2k}$ or $\mathcal{H}_{2k}^-$.

Let $\mathcal{H}_{2k-1} = \mathcal{H}_{2k}(x_1, \ldots, x_{2k-1}, 0)$, which is obtained from $\mathcal{H}_{2k}$ by replacing $x_{2k}$ by $0$.
It's easy to see our canonical form for HOD is well defined, that is $\mathcal{H}^\pm_{2k}$ and $\mathcal{H}_{2k-1}$ are square HODs with size $[2^{k-1}, 2^{k-1}, k]$.

\vspace{0.2cm}
Let's restate our main theorem in the context of square HOD.

\begin{theorem} Square HOD $\mathcal{O}_x$ with size $[n, n, 2k]$ exists if and only if
$$
2^{k-1} | n,
$$
and there exist $U, V \in U_n(\mathbb{C})$, and unique $n_1, n_2 \in \mathbb{N}$ with $n_1 +n_2 = n/2^{k-1}$ such that
\begin{eqnarray*}
\mathcal{O}_x & = & U  \diag(\underbrace{\mathcal{H}_{2k}, \ldots, \mathcal{H}_{2k}}_{n_1}, \underbrace{\mathcal{H}^{-}_{2k}, \ldots, \mathcal{H}^{-}_{2k}}_{n_2}) V \\
& = & U \left((I_{n_1} \otimes \mathcal{H}_{2k}) \oplus (I_{n_2} \otimes \mathcal{H}^-_{2k}) \right) V.
\end{eqnarray*}
\end{theorem}

Following is the structure theorem for square HOD with odd number of variables, which is very similar to
the even case except there is only one canonical form. The reason behind is that there is only one nondegenerate
irreducible representation for the corresponding group.

\begin{theorem}
\label{thm:SHODodd}

Square HOD $\mathcal{O}_x$ with size $[n, n, 2k-1]$ exists if and only if
$$
2^{k-1} | n,
$$
and there exist $U, V \in U_n(\mathbb{C})$ such that
\begin{eqnarray*}
\mathcal{O}_x & = & U \diag(\mathcal{H}_{2k-1}, \mathcal{H}_{2k-1}, \ldots, \mathcal{H}_{2k-1}) V \\
& = & U \left(I_{n/2^{k-1}} \otimes \mathcal{H}_{2k-1} \right) V.
\end{eqnarray*}
\end{theorem}
\begin{remark} As a consequence, any two square HODs with size $[n, n, 2k-1]$ are equivalent, since there is only one canonical form.
\end{remark}
\begin{remark} If $n = 2^{k-1}$, $U, V$ are uniquely determined by Schur's lemma. Otherwise, $U, V$ are not unique.
\end{remark}

The proof is almost the same, and we will sketch the idea and omit the details. Square HOD $\mathcal{O}_x$ with size
$[n, n, 2k-1]$ is equivalent to the existence of $2k-1$ unitary matrices $E_0, \ldots, E_{2k-2} \in U_n(\mathbb{C})$ such that
$E_i^H E_j + E_j^H E_i = 0$ for all $i \not= j$. By the normalization trick, let $G_i = E_0^H E_i$ for $i = 1, 2, \ldots, 2k-2$.
It is easily checked that $G_i$ are unitary matrices satisfying $G_i^2 = -I$ and $G_i G_j = -G_j G_i$, which induces
a unitary representation of group $\mathcal{G}_{2k-2}$, where finite group $\mathcal{G}_{2k-2}$ is generated by $g_1, \ldots, g_{2k-2}, 1, -1$
satisfying $g_i^2 = -1$ and $g_i g_j = -g_j g_i$ for all $i \not= j$. It turns out this group has $2^{2k-2}+1$ irreducible representations,
one is $2^{k-1}$ dimensional, $2^{2k-2}$ are one dimensional (see Lemma \ref{lem:G2k}).

There is a proof of the following lemma in \cite{TH02}. For completeness, we reproduce the proof in the appendix.

\begin{lemma} \cite{TH02} Group $\mathcal{G}_{2k}$ has $2^{2k}+1$ irreducible representations. One is $2^k$ dimensional, and $2^{2k}$ are one dimensional.
\label{lem:G2k}
\end{lemma}

To finish the proof, we need the following lemma, which says our canonical form induces the nondegenerate irreducible of dimension $2^{k-1}$ representation
of group $\mathcal{G}_{2k-2}$.

\begin{lemma} Square HOD $\mathcal{H}_{2k-1}$ induces a $2^{k-1}$ dimensional irreducible representation of group $\mathcal{G}_{2k-2}$.
\end{lemma}
\begin{proof} As usual, write
$$
\mathcal{H}_{2k-1} = \sum_{i=1}^{2k-1} E_{i-1} x_i,
$$
and let $G_i = E_0^H E_i$ for $i = 1, \ldots, 2k-2$. We have already seen map $\rho: \mathcal{G}_{2k-2} \to U_n(\mathbb{C})$, defined by $\rho(g_i) = G_i$, is a unitary representation of group $\mathcal{G}_{2k-2}$ since $G_i^2 = -1$ and
$G_i G_j = -G_j G_i$ for all $i \not= j$.

By Lemma \ref{lem:G2k}, there are two possibilities: either $\rho$ is the irreducible representation of dimension $2^{k-1}$, or $\rho$ is a direct sum
of one dimensional representations. Observe that one dimensional representations are all degenerate, because assume $\pi : \mathcal{G}_{2k-2} \to \mathbb{C}$
is a one dimensional, then $\pi(g_i) \pi(g_j) = \pi(g_i g_j) = \pi(-g_j g_i) = \pi(-1) \pi(g_j) \pi(g_i)$, which implies $\pi(-1) = 1$.
Thus, if $\rho$ is a direct sum of one dimensional representations, then $\rho$ is also degenerate, which is a contradiction! Therefore, $\rho$ must
be the irreducible one.
\end{proof}

Let's go back to the proof of Theorem \ref{thm:SHODodd}.
Since square HOD $\mathcal{O}_x$ with size $[n, n, 2k-1]$ induces an unitary representation $\rho$ of group $\mathcal{G}_{2k-2}$, by classical result in representation theory, it is a direct sum of irreducible ones. Since $\rho$ is unregenerate, i.e., $\rho(-1) = -I$, and the only unregenerate irreducible representation of group $\mathcal{G}_{2k-2}$ is of $2^{k-1}$ dimensional, $n$, the dimensional of $\rho$,  must be a multiple of $2^{k-1}$, which is the dimension of
the irreducible representation induced by canonical form $\mathcal{H}_{2k-1}$, which proves the first part of the theorem. Following the same argument in the proof of our main theorem, that writing $\mathcal{O}_x$ explicitly as a direct sum of irreducible ones induced by the canonical form and apply Lemma \ref{lem:unitary_equ}, Theorem \ref{thm:SHODodd} is proved.

\vspace{0.2cm}
As a consequence of the Theorem \ref{thm:SHODodd}, square HOD with size $[n, n, 2k-1]$
can be ``extended'' to a square HOD with size $[n, n, 2k]$, which is equivalent to a square COD with size $[n, n, k]$ as
we have already seen.

\begin{corollary} Assume $\mathcal{O}_x$ is a square HOD with size $[n, n, 2k-1]$. There exists a matrix $\mathcal{L}_x$ where
each entry is a $\mathbb{C}$ linear combination of $x_{2k}$ such that
$
\mathcal{O}_x + \mathcal{L}_x
$
is a square HOD with size $[n, n, 2k]$, or equivalently, a square COD with size $[n, n, k]$ by setting $x_{2i-1} = (z_i + z^*_{i})/2$
and $x_{2i} = (z_i - z^*_i) / (2\sqrt{-1})$.
\end{corollary}
\begin{proof} By Theorem \ref{thm:SHODodd}, $\mathcal{O}_x$ is equivalent to the canonical form $\mathcal{H}_{2k-1}$, that is,
$$
\mathcal{O}_x = U \left( \bigoplus_{i=1}^{n/2^{k-1}}\mathcal{H}_{2k-1} \right) V,
$$
where $U, V$ are unitary matrices of size $n$.

Let
\begin{equation}
\label{equ:defL}
\mathcal{L}_x = U \left( \bigoplus_{i=1}^{n/2^{k-1}}\mathcal{H}_{2k-1}(0, \ldots, 0, \pm x_{2k}) \right) V,
\end{equation}
 where $\mathcal{H}_{2k-1}(0, \ldots, 0, \pm x_{2k})$ denotes the matrix obtained from $\mathcal{H}_{2k-1}$
by replacing $x_1, \ldots, x_{2k-1}$ by $0$, and possibly replacing $x_{2k}$ by $-x_{2k}$.
Then
\begin{eqnarray*}
\mathcal{O}_x + \mathcal{L}_x & =  & U \left(\bigoplus_{i=1}^{n/2^{k-1}} \mathcal{H}_{2k-1} + \mathcal{H}_{2k}(0, \ldots, 0, \pm x_{2k}) \right) V \\
& = & U \left(\bigoplus_{i=1}^{n/2^{k-1}}\mathcal{H}_{2k}(x_1, \ldots, x_{2k-1}, 0) + \mathcal{H}_{2k}(0, \ldots, 0, \pm x_{2k})\right) V \\
& = & U \left(\bigoplus_{i=1}^{n/2^{k-1}}\mathcal{H}^\pm_{2k}\right) V,
\end{eqnarray*}
is $[n, n, 2k]$ square HOD, where $\mathcal{H}_{2k-1} = \mathcal{H}_{2k}(x_1, \ldots, x_{2k-1}, 0)$ by definition.
\end{proof}
\begin{remark}
$\mathcal{L}_x$ in \eqref{equ:defL} are the only possibilities to ``extend'' $\mathcal{O}_z$, since $\mathcal{G}_{2k-1}$ has only
two irreducible nondegenerate representations, and all nondegenerate representations are the direct sum of these two.
\end{remark}


\section{Conclusion and Open Problems}

Square CODs can be completely understood by group representation approach. The high level idea is quite clear and 
general: if you are interested in some mysterious object $\mathcal{O}$, assume $\mathcal{O}$ exists, then it will induce a representation
of some group with certain properties. By studying the group, you will understand its representations, and thus understand the mysterious object $\mathcal{O}$
hopefully. For square CODs, it seems that  everything is clear now.

For further research, it's tempting to apply the group representation approach for the nonsquare complex orthogonal
if possible. Another open problem is to apply similar approach to quasiorthogonal designs, that is, only some given pairs
of columns are orthogonal.

\section*{Acknowledgement}
We would like to thank Mateusz Wasilewski for giving the proof of Lemma \ref{lem:unitary_equ} on MathOverflow.

\section*{Appendix}
\begin{lemma} \cite{TH02}
For group $\mathcal{G}_{2k-1}$, if $\rho$ is an irreducible representation with dimension greater than $1$, then
$\pi$ defined by $\pi(g_i) = \rho(g_i)$ for all $i \not= 2k-1$, and $\pi(g_{2k-1}) = -\rho(g_{2k-1})$ is another $2^{k-1}$ dimensional
irreducible
representation.
\end{lemma}
\begin{proof} First, notice that $\prod_{i=1}^{2k-1} g_i$ is a central element of the group. Because for any $S \subseteq \{1, \ldots, 2k-1\}$,
where $S= \{s_1, \ldots, s_l \}$,
\begin{eqnarray*}
\left( \prod_{i=1}^{2k-1} g_i \right) \left( \prod_{i \in S} g_i\right) & = &
(-1)^{2k-2} g_{s_1} \left( \prod_{i=1}^{2k-1} g_i \right) \left( \prod_{i \in S \setminus \{s_1\}} g_i\right)\\
& = &  g_{s_1} g_{s_2} \left( \prod_{i=1}^{2k-1} g_i \right) \left( \prod_{i \in S \setminus \{s_1, s_2\}} g_i\right) \\
& = &  \left( \prod_{i \in S} g_i\right) \left( \prod_{i=1}^{2k-1} g_i \right).
\end{eqnarray*}
By Schur's lemma, since $\rho$ is irreducible, $\rho(\prod_{i=1}^{2k-1} g_i) = \lambda I_n$, where $\lambda \in \mathbb{C}$, which
implies
$$
\rho(g_1 \ldots g_{2k-2}) = \rho(-g_{2k-1}^2) \rho(g_1 \ldots g_{2k-2}) = \rho(-g_{2k-1}) \rho(g_1 \ldots g_{2k-1}) = -\rho(g_{2k-1}) \lambda I_n,
$$
i.e.,
\begin{equation}
\label{equ:repg1}
\rho(g_{2k-1}) = -\frac{1}{\lambda} \rho(g_1 \ldots g_{2k-2}).
\end{equation}
Assume to the contrary that there exists a similarity transformation $T \in GL_n(\mathbb{C})$ such that $\pi = T^{-1} \rho T$. By
the definition of $\pi$, we have
$$
\pi(g_i) = T^{-1} \rho(g_i) T
$$
for all $i = 1, \ldots, 2k-2$. And
\begin{eqnarray*}
\pi(g_{2k-1}) & = & T^{-1} \rho(g_{2k-1}) T \\
& = & T^{-1} (-\frac{1}{\lambda} \rho(g_1) \rho(g_2) \ldots \rho(g_{2k-2})) T \\
& = & -\frac{1}{\lambda} (T^{-1} \rho(g_1) T) (T^{-1} \rho(g_2) T) \ldots (T^{-1} \rho(g_{2k-2}) T)\\
& = & -\frac{1}{\lambda} \pi(g_1) \ldots \pi(g_{2k-2})\\
& = & \rho(g_{2k-1}),
\end{eqnarray*}
which is a contradiction with the definition of $\pi$!
\end{proof}

\begin{lemma}
\cite{TH02}
Group $\mathcal{G}_{2k-1}$ has $2^{2k-1}+2$ irreducible representations. Two are $2^{k-1}$ dimensional,
$2^{2k-1}$ are 1-dimensional.
\end{lemma}
\begin{proof} First, let's construct $2^{2k-1}$ nonequivalent one dimensional representations. For any $J \subseteq \{1, \ldots, 2k-1\}$,
let $\rho(1) = \rho(-1) = 1$, $\rho(g_i) = 1$ if $i \not\in J$ and $\rho(g_j) = -1$. It's easy to see $\rho$ is a representation
of $\mathcal{G}_{2k-1}$, and they are nonequivalent.

Then, apply the counting formula, e.g. section 2.4 in \cite{Serre},
$$
|\mathcal{G}_{2k-1}| = \sum_{i=1}^{l} n_i^2,
$$
where $n_i$ is the dimension of each irreducible representations, and $l$ equals the number of conjugacy classes. We claim $l = 2^{2k-1}+2$,
which will be proved at the end of this proof. By the existence of $2^{2k-1}$ one dimensional representations, and Lemma \ref{lem:irr_rep}, we have
$$2n_1^2 = 2^{2k-1},$$
which implies $n_1 = 2^{k-1}$, i.e., there exist two irreducible two representations with dimension $2^{k-1}$.

Finally, we need to prove our claim: there are $2^{2k-1}+2$ conjugacy classes. If an element commute with all
elements, then itself forms a conjugacy class; otherwise, itself and its negation forms a conjugacy class. For
$\pm 1, \pm \prod_{i=1}^{2k-1}{g_i}$, they belong to the former case; for any $\emptyset \not= S \subsetneq [2k-1]$, with $S=\{s_1, \ldots, s_m\}$, there exists $s'_m \not\in S$, it's easy to verify $\pm \prod_{i \in S} g_i$ is anti-commuting with $(\prod_{i=1}^{m-1} s_i) s'_m$. Therefore, there are $2^{2k-1}+2$ conjugacy classes,
which completes our proof.
\end{proof}

\begin{lemma} \cite{TH02} Group $\mathcal{G}_{2k}$ has $2^{2k}+1$ irreducible representations. One is $2^k$ dimensional, and $2^{2k}$ are one dimensional.
\end{lemma}
\begin{proof} The idea is to apply the counting formula, that, the number of irreducible representations equals the number of conjugacy classes, and
the sum of square of dimensions for each irreducible representations equals the size of the group.

For group  $\mathcal{G}_{2k}$, both elements $1$ and $-1$ form a conjugacy class by itself, since they commute with all other elements. For any $\emptyset \neq S \subseteq [2k]$, the elements $\prod_{i \in S} g_i$ and $-\prod_{i \in S} g_i$ form a conjugacy class. To see this, we discuss by cases. \textbf{Case 1:} $S = [2k]$. It's easily verified $g_1 \left(\prod_{i \in S} g_i\right) g_1^{-1} = - \prod_{i \in S} $, which proves $\prod_{i \in S}$ and $-\prod_{i \in S}$ form
a conjugacy class.   \textbf{Case 2:} $S \subsetneq [2k]$ and $|S|$ is odd. Take $i \in [2k] \setminus S$. It's easily verified $g_i \left(\prod_{i \in S} g_i\right) g_1^{-1} = (-1)^{|S|} \prod_{i \in S} = - \prod_{i \in S}$. \textbf{Case 3:} $S \subsetneq [2k]$ and $|S|$ is even. Take any $i \in S$. It's easily verified $g_i \left(\prod_{i \in S} g_i\right) g_1^{-1} = (-1)^{|S|-1} \prod_{i \in S} = - \prod_{i \in S}$. Hence, we conclude there are $2^{2k}+1$ conjugacy
classes.

For any $S \subset [2k]$, define $\rho : \mathcal{G}_{2k} \to \mathbb{C}$ such that $\rho(g_i) = -1$ if $i \in S$, otherwise $1$, which is an irreducible
 one dimensional representation. Therefore, there exists $2^k$ one dimensional representations. Since there are $2^k+1$ conjugacy classes, there remains only
 one irreducible representations, and the dimension is $\sqrt{|\mathcal{G}_{2k}| - 2^{2k}} = 2^k$ by counting formula.
\end{proof}


\begin{thebibliography}{99}

\bibitem{Adams62}
J.~F.~Adams, \emph{``Vector fields on spheres,''} Ann. Math., vol. 75, no.
2, pp. 603-632, 1962.

\bibitem{ALP65} J. F. Adams, P. D. Lax, and R. S. Phillips, \emph{``On matrices whose real linear
combinations are nonsingular,''} in Proc. Amer. Math. Soc., vol. 16, 1965,
pp. 318-322.

\bibitem{ALP65_2} J. F. Adams, P. D. Lax, and R. S. Phillips, \emph{``Corrections to 'On matrices whose real linear combinations are nonsingular',''}
in Proc. Amer. Math. Soc., vol. 17, 1966, pp. 945-947.

\bibitem{ADK11}
S.~S.~Adams, J.~Davis, N.~Karst, M.~K.~Murugan, B.~Lee, M.~Crawford,
C.~Greeley, \emph{``Novel classes of minimal delay and low PAPR rate 1/2
complex orthogonal designs,''} IEEE Trans. Inf. Theory, vol. 57, no.
4, pp. 2254-2262, Apr. 2011.

\bibitem{AKM11}
S.~S.~Adams, N.~Karst, M.~K.~Murugan and T.~A.~Wysocki, \emph{``On
transceiver signal linearization and the decoding delay of maximum
rate complex orthogonal space-time block codes,''} IEEE Trans. Inf.
Theory, vol. 57, no.  6, pp. 3618-3621, Jun. 2011.

\bibitem{AKP07}
S.~S.~Adams, N.~Karst, and J.~Pollack, \emph{``The minimum decoding delay
of maximum rate complex orthogonal space-time block codes,''} IEEE
Trans. Inf. Theory, vol. 53, no. 8, pp. 2677-2684, Aug. 2007.

\bibitem{AKM10}
S.~S.~Adams, N.~Karst, M.~K.~Murugan., \emph{``The final case of the
decoding delay problem for maximum rate complex orthogonal
designs,''} IEEE Trans. Inf. Theory, vol. 56, no. 1, pp. 103-122,
Jan. 2010.

\bibitem{Ala98}
S.~Alamouti, \emph{``A simple transmit diversity technique for wireless
communications,''} IEEE J. Select. Areas Commun., vol. 16, pp.
1451-1458, Oct. 1998.

\bibitem{DR12}
S.~Das and B.~S.~Rajan, \emph{``Low-delay, high-rate nonsquare complex orthogonal designs,''}
IEEE Trans. on Information Theory, vol.~58, no.~5, pp. 2633-2647, May, 2012.

\bibitem{GP79}
A.~V.~Geramite and N.~J.~Pullman, \emph{``Orthogonal Designs: Quadratic
Forms and Hadamard Matrices (Lecture Notes in Pure and Applied
Mathematics),''} New York: Marcel Dekker, vol. 43, 1979.

\bibitem{TH02}
O.~Tirkkonen, and A.~Hottinen, ``\emph{Square-matrix embeddable space-time block codes for complex signal constellations,}'' IEEE Trans.
on Information Theory, vol. 48, no. 2, Feb. 2002.

\bibitem{KS05} H.~Kan and H.~Shen, \emph{``A counterexample for the open problem
on the minimal delays of orthogonal designs with maximal rates,''}
IEEE Trans. Inf. Theory, vol. 51, no. 1, pp. 355-359, Jan. 2005.


\bibitem{Lia03} X-B. Liang, \emph{``Orthogonal designs with maximal rates,''} IEEE
Trans. Inf. Theory, vol. 49, no. 10, pp. 2468-2503, Oct. 2003.

\bibitem{LK12}
Y.~Li and H.~Kan, \emph{``Complex orthogonal designs with forbidden $2 \times 2$ submatrices''}, IEEE Trans. on Information Theory, vol.~58, no.~7, pp.~4825--4836, July, 2012.

\bibitem{LFX05} K. Lu, S. Fu, and X. G. Xia, \emph{``Closed-form designs of
complex orthogonal space-time block codes of rates $(k+1)/(2k)$ for
$2k-1$ or $2k$ transmit antennas,''} IEEE Trans. Inf. Theory, vol. 51,
no. 12, pp. 4340-4347, Dec. 2005.

\bibitem{Serre}
J.-P.~Serre, \emph{Linear Representations of Finite Groups,} Springer-Verlag, 2008.

\bibitem{Shapiro}
D.~B.~Shapiro, \emph{Compositions of Quadratic Forms,} De Gruyter expositions in mathematics, 2000.

\bibitem{SXL04} W. Su, X.-G. Xia, and K. J. R. Lui, \emph{``A systematic design of
high-rate complex orthogonal space-time block codes,''} IEEE Commun.
Lett. vol. 8, no. 6, pp. 380-382, Jun. 2004.

\bibitem{SSW05} J. Seberry, S. A. Spence, and T. A. Wysocki, \emph{``A construction
technique for generalized complex orthogonal designs and
applications to wireless communications,''} Linear Algebra and its
Applications 405 (2005) 163-167.

\bibitem{TJC99} V. Tarokh, H. Jafarkhani, and A. R. Calderbank, \emph{``Space-time
block codes from orthogonal designs,''} IEEE Trans. Inf. Theory, vol.
45, no. 5, pp. 1456-1467, July 1999.

\bibitem{TJC00} V. Tarokh, H. Jafarkhani, and A. R. Calderbank, \emph{``Correction
to `Space-time block codes from orthogonal designs,''} IEEE Trans.
Inf. Theory, vol. 46, no. 1, Jan. 2000.

\bibitem{TWS09} L. C. Tran, T. A. Wysocki, J. Seberry, A. Mertins and S. S.
Adams, \emph{``Novel constructions of improved square complex orthogonal
designs for eight transmit antennas,''} IEEE Trans. Inf. Theory, vol.
55, no. 10, Oct. 2009.

\bibitem{WX03} H. Wang and X-G. Xia, \emph{``Upper bounds of rates of complex
orthogonal space-time block codes,''} IEEE Trans. Inf. Theory, vol.
49, no. 10, pp. 2788-2796, Oct. 2003.

\end{thebibliography}
\end{document}